\documentclass[aps,notitlepage,twocolumn,nofootinbib,superscriptaddress,longbibliography]{revtex4-2}

\usepackage[dvipsnames]{xcolor}
\usepackage{framed}
\definecolor{shadecolor}{rgb}{0.9,0.9,0.9}

\usepackage{mathtools}
\usepackage{amsmath}
\usepackage[shortlabels]{enumitem}

\usepackage{graphicx,epic,eepic,epsfig,amsmath,latexsym,amssymb,verbatim,color}
 
\usepackage{amsfonts}       
\usepackage{nicefrac}       

\usepackage{amsmath}
\usepackage{bbm}

\usepackage{float}
\usepackage{tikz}
\usetikzlibrary{chains}
\usetikzlibrary{fit}
\usepackage{pgflibraryarrows}		
\usepackage{pgflibrarysnakes}		

\usepackage{epsfig}
\usetikzlibrary{shapes.symbols,patterns} 
\usepackage{pgfplots}

\usepackage[strict]{changepage}
\usepackage{hyperref}
\hypersetup{colorlinks=true,citecolor=blue,linkcolor=blue,filecolor=blue,urlcolor=blue,breaklinks=true}

\usepackage[marginal]{footmisc}
\usepackage{url}
\usepackage{theorem}

\newtheorem{definition}{Definition}
\newtheorem{proposition}{Proposition}
\newtheorem{lemma}[proposition]{Lemma}

\newtheorem{theorem}[proposition]{Theorem}

\newtheorem{corollary}[proposition]{Corollary}

\def\squareforqed{\hbox{\rlap{$\sqcap$}$\sqcup$}}
\def\qed{\ifmmode\squareforqed\else{\unskip\nobreak\hfil
\penalty50\hskip1em\null\nobreak\hfil\squareforqed
\parfillskip=0pt\finalhyphendemerits=0\endgraf}\fi}
\def\endenv{\ifmmode\;\else{\unskip\nobreak\hfil
\penalty50\hskip1em\null\nobreak\hfil\;
\parfillskip=0pt\finalhyphendemerits=0\endgraf}\fi}
\newenvironment{proof}{\noindent \textbf{{Proof~} }}{\hfill $\blacksquare$}

\newcounter{remark}

\newcounter{example}
\newenvironment{example}[1][]{\refstepcounter{example}\par\medskip\noindent%
\textbf{Example~\theexample #1} }{\medskip}

\newcommand{\dbraket}[2]{\langle \hspace{-.8mm} \langle #1 \vert \hspace{-.4mm} #2 \rangle\hspace{-.8mm}\rangle}

\newcommand{\dket}[1]{\vert #1 \rangle \hspace{-.8mm} \rangle}

\mathchardef\ordinarycolon\mathcode`\:
\mathcode`\:=\string"8000
\def\vcentcolon{\mathrel{\mathop\ordinarycolon}}
\begingroup \catcode`\:=\active
  \lowercase{\endgroup
  \let :\vcentcolon
  }

\usepackage{cleveref}
\usepackage{graphicx}
\usepackage{xcolor}

\RequirePackage[framemethod=default]{mdframed}
\newmdenv[skipabove=7pt,
skipbelow=7pt,
backgroundcolor=darkblue!15,
innerleftmargin=5pt,
innerrightmargin=5pt,
innertopmargin=5pt,
leftmargin=0cm,
rightmargin=0cm,
innerbottommargin=5pt,
linewidth=1pt]{tBox}

\newmdenv[skipabove=7pt,
skipbelow=7pt,
backgroundcolor=red!15,
innerleftmargin=5pt,
innerrightmargin=5pt,
innertopmargin=5pt,
leftmargin=0cm,
rightmargin=0cm,
innerbottommargin=5pt,
linewidth=1pt]{rBox}

\newmdenv[skipabove=7pt,
skipbelow=7pt,
backgroundcolor=blue2!25,
innerleftmargin=5pt,
innerrightmargin=5pt,
innertopmargin=5pt,
leftmargin=0cm,
rightmargin=0cm,
innerbottommargin=5pt,
linewidth=1pt]{dBox}
\newmdenv[skipabove=7pt,
skipbelow=7pt,
backgroundcolor=darkkblue!15,
innerleftmargin=5pt,
innerrightmargin=5pt,
innertopmargin=5pt,
leftmargin=0cm,
rightmargin=0cm,
innerbottommargin=5pt,
linewidth=1pt]{sBox}
\definecolor{darkblue}{RGB}{0,76,156}
\definecolor{darkkblue}{RGB}{0,0,153}
\definecolor{blue2}{RGB}{102,178,255}
\definecolor{darkred}{RGB}{195,0,0}

\newcommand{\nc}{\newcommand}
\nc{\rnc}{\renewcommand}
\nc{\lbar}[1]{\overline{#1}}
\nc{\bra}[1]{\langle#1|}
\nc{\ket}[1]{|#1\rangle}
\nc{\ketbra}[2]{|#1\rangle\!\langle#2|}
\nc{\braket}[2]{\langle#1|#2\rangle}

\nc{\proj}[1]{| #1\rangle\!\langle #1 |}
\nc{\avg}[1]{\langle#1\rangle}
\nc{\smfrac}[2]{\mbox{$\frac{#1}{#2}$}}
\nc{\tr}{\operatorname{Tr}}
\nc{\sign}{\operatorname{sign}}
\nc{\ox}{\otimes}
\nc{\dg}{\dagger}
\nc{\dn}{\downarrow}
\nc{\cA}{{\cal A}}
\nc{\cB}{{\cal B}}
\nc{\cC}{{\cal C}}
\nc{\cD}{{\cal D}}
\nc{\cE}{{\cal E}}
\nc{\cF}{{\cal F}}
\nc{\cG}{{\cal G}}
\nc{\cH}{{\cal H}}
\nc{\cI}{{\cal I}}
\nc{\cJ}{{\cal J}}
\nc{\cK}{{\cal K}}
\nc{\cL}{{\cal L}}
\nc{\cM}{{\cal M}}
\nc{\cN}{{\cal N}}
\nc{\cO}{{\cal O}}
\nc{\cP}{{\cal P}}
\nc{\cQ}{{\cal Q}}
\nc{\cR}{{\cal R}}
\nc{\cS}{{\cal S}}
\nc{\cT}{{\cal T}}
\nc{\cU}{{\cal U}}
\nc{\cV}{{\cal V}}
\nc{\cX}{{\cal X}}
\nc{\cY}{{\cal Y}}
\nc{\cZ}{{\cal Z}}
\nc{\cW}{{\cal W}}
\nc{\csupp}{{\operatorname{csupp}}}
\nc{\qsupp}{{\operatorname{qsupp}}}
\nc{\var}{{\operatorname{var}}}
\nc{\rar}{\rightarrow}
\nc{\lrar}{\longrightarrow}
\nc{\polylog}{{\operatorname{polylog}}}
\nc{\wt}{{\operatorname{wt}}}
\nc{\av}[1]{{\left\langle {#1} \right\rangle}}
\nc{\supp}{{\operatorname{supp}}}

\nc{\argmin}{{\operatorname{argmin}}}

\def\x{\xi}

\nc{\RR}{{{\mathbb R}}}
\nc{\CC}{{{\mathbb C}}}
\nc{\FF}{{{\mathbb F}}}
\nc{\NN}{{{\mathbb N}}}
\nc{\ZZ}{{{\mathbb Z}}}
\nc{\PP}{{{\mathbb P}}}
\nc{\QQ}{{{\mathbb Q}}}
\nc{\UU}{{{\mathbb U}}}
\nc{\EE}{{{\mathbb E}}}
\nc{\id}{{\operatorname{id}}}

\nc{\CHSH}{{\operatorname{CHSH}}}

\nc{\be}{\begin{equation}}
\nc{\ee}{{\end{equation}}}
\nc{\bea}{\begin{eqnarray}}
\nc{\eea}{\end{eqnarray}}
\nc{\<}{\langle}
\rnc{\>}{\rangle}
\nc{\rU}{\mbox{U}}

\nc{\ob}[1]{#1}

\nc{\SEP}{{\text{\rm SEP}}}
\nc{\NS}{{\text{\rm NS}}}
\nc{\LOCC}{{\text{\rm LOCC}}}
\nc{\PPT}{{\text{\rm PPT}}}
\nc{\EXT}{{\text{\rm EXT}}}
\nc{\Sym}{{\operatorname{Sym}}}

\nc{\ERLO}{{E_{\text{r,LO}}}}
\nc{\ERLOCC}{{E_{\text{r,LOCC}}}}
\nc{\ERPPT}{{E_{\text{r,PPT}}}}
\nc{\ERLOCCinfty}{{E^{\infty}_{\text{r,LOCC}}}}
\nc{\Aram}{{\operatorname{\sf A}}}

\usepackage{tikz}
\usepackage{hyperref}
\hypersetup{colorlinks=true,citecolor=blue,linkcolor=blue,filecolor=blue,urlcolor=blue,breaklinks=true}

\makeatletter
\def\grd@save@target#1{%
  \def\grd@target{#1}}
\def\grd@save@start#1{%
  \def\grd@start{#1}}
\tikzset{
  grid with coordinates/.style={
    to path={%
      \pgfextra{%
        \edef\grd@@target{(\tikztotarget)}%
        \tikz@scan@one@point\grd@save@target\grd@@target\relax
        \edef\grd@@start{(\tikztostart)}%
        \tikz@scan@one@point\grd@save@start\grd@@start\relax
        \draw[minor help lines,magenta] (\tikztostart) grid (\tikztotarget);
        \draw[major help lines] (\tikztostart) grid (\tikztotarget);
        \grd@start
        \pgfmathsetmacro{\grd@xa}{\the\pgf@x/1cm}
        \pgfmathsetmacro{\grd@ya}{\the\pgf@y/1cm}
        \grd@target
        \pgfmathsetmacro{\grd@xb}{\the\pgf@x/1cm}
        \pgfmathsetmacro{\grd@yb}{\the\pgf@y/1cm}
        \pgfmathsetmacro{\grd@xc}{\grd@xa + \pgfkeysvalueof{/tikz/grid with coordinates/major step}}
        \pgfmathsetmacro{\grd@yc}{\grd@ya + \pgfkeysvalueof{/tikz/grid with coordinates/major step}}
        \foreach \x in {\grd@xa,\grd@xc,...,\grd@xb}
        \node[anchor=north] at (\x,\grd@ya) {\pgfmathprintnumber{\x}};
        \foreach \y in {\grd@ya,\grd@yc,...,\grd@yb}
        \node[anchor=east] at (\grd@xa,\y) {\pgfmathprintnumber{\y}};
      }
    }
  },
  minor help lines/.style={
    help lines,
    step=\pgfkeysvalueof{/tikz/grid with coordinates/minor step}
  },
  major help lines/.style={
    help lines,
    line width=\pgfkeysvalueof{/tikz/grid with coordinates/major line width},
    step=\pgfkeysvalueof{/tikz/grid with coordinates/major step}
  },
  grid with coordinates/.cd,
  minor step/.initial=.2,
  major step/.initial=1,
  major line width/.initial=2pt,
}
\makeatother

\usepackage{thmtools}
\usepackage{thm-restate}
\usepackage{etoolbox}
\makeatletter
\def\problem@s{}
\newcounter{problems@cnt}

\newcommand{\allproblems}{\problem@s}
\makeatother

\usepackage{amsmath,amsfonts,amssymb,array,graphicx,mathtools,multirow,bm,times,tcolorbox,relsize,booktabs}
\usepackage[utf8]{inputenc}
\usepackage[T1]{fontenc}
\usepackage{tikz} 
\usetikzlibrary{quantikz2}
\usepackage[ruled, vlined, linesnumbered]{algorithm2e}

\usepackage[vcentermath, enableskew]{youngtab}
\Yboxdim{5pt}

\definecolor{colortwo}{rgb}{0.4,0.77,0.17}
\definecolor{colorthree}{rgb}{0.01,0.51,0.93}

\newcommand{\update}[1]{\textcolor{black}{#1}}

\nc{\EPPT}{{E_{\operatorname{PPT}}}}
\nc{\EPPTone}{{E_{\operatorname{PPT}}^{(1)}}}
\nc{\EK}{{E_{\kappa}}}
\allowdisplaybreaks

\begin{document}
\title{Conclusive exclusion of quantum states with group action}
\author{Hongshun Yao}
\email{yaohongshun2021@gmail.com}
\author{Xin Wang}
\email{felixxinwang@hkust-gz.edu.cn}
\affiliation{Thrust of Artificial Intelligence, Information Hub,\\
The Hong Kong University of Science and Technology (Guangzhou), Guangdong 511453, China}

\date{\today}

\begin{abstract}
Retrieving classical information from quantum systems is central to quantum information processing. As a more general task than quantum state discrimination, which focuses on identifying the exact state, quantum state exclusion only requires ruling out options, revealing fundamental limits of information extraction from quantum systems. 
\update{In this work, we study the conclusive exclusion of quantum states generated by group actions, establishing explicit criteria for when such exclusion is possible. For systems with complex symmetries, including finite and compact Lie groups, we derive a sufficient condition for conclusive exclusion based on the initial state's amplitudes and the group's structure.} As applications to special groups such as Abelian groups, we establish necessary and sufficient conditions for conclusive state exclusion and generalize the Pusey-Barrett-Rudolph result to a wider range of scenarios. Finally, we explore zero-error communication via conclusive exclusion of quantum states and derive a lower bound on the feedback-assisted and non-signalling-assisted zero-error capacity of classical-quantum channels generated by group actions.
\end{abstract}

\maketitle

\section{Introduction}\label{sec:introduction}
Extracting information from quantum systems lies at the heart of quantum information theory, with quantum state discrimination being the fundamental model~\cite{barnett2009quantum,Ghosh2001,bae2015quantum,watrous2018theory,Rudolph2003,Chitambar2014b,WW19,Chitambar2014b,Schmid2017,Rosati2017,Takagi2018b,Leditzky2022b,Lami2017,Zhu2024,Zhu2025}. To extract information from a quantum system with several possible states, quantum measurements are performed to determine the system's state. Perfectly discriminating two possible states reveals one bit of information, which happens when the two states are orthogonal. However, it is well-known that nonorthogonal quantum states cannot be distinguished perfectly, even when a large but finite number of copies are available. This limitation raises the question: when state discrimination cannot be effectively performed, what can we conclusively deduce about the quantum system?

A more general and less ambitious task is to reveal only partial but conclusive information about the system, such as excluding possible state options of the quantum system via measurements. This task is known as conclusive state exclusion~\cite{pusey2012reality,bandyopadhyay2014conclusive,Uola2020,mishra2024optimal,McIrvin2024}. If we can rule out at least one state with certainty based on each measurement outcome, we say that conclusive state exclusion is successful. Conclusive state exclusion represents the most fundamental aspect of state discrimination theory, as it establishes the basic framework for drawing conclusive information of the quantum system.

Conclusive state exclusion has gained importance in the context of the foundations of quantum mechanics due to a result by Pusey, Barrett, and Rudolph (PBR) \cite{pusey2012reality}. Specifically, conclusive single-state exclusion was introduced to demonstrate a no-go theorem in quantum theory, stating that no local theory can reproduce the predictions of quantum mechanics. Caves et al.~\cite{Caves2002a} derived the necessary and sufficient conditions for single-state conclusive exclusion of a set of three pure states, and the problem has since been formulated within the framework of semidefinite programs. 
\update{However, for general ensembles of quantum states, deriving explicit, operational conditions on the states and measurements that guarantee conclusive exclusion remains a challenging problem, which limits our understanding of the fundamental bounds of information extraction from quantum systems.}

As conclusive state exclusion characterizes the fundamental limit of extracting exact information from quantum systems, it is fundamentally important to zero-error information theory, particularly from the perspective of determining when a quantum channel has positive zero-error capacity. Shannon described the zero-error capacity of a channel as the maximum rate at which it can transmit information with perfect reliability. The zero-error capacity of quantum channels has been extensively studied over the past two decades~\cite{Cubitt2010,Cubitt2011,Cubitt2011a,Duan2015a,Leung2015,Duan2013,Duan2016,Duan2015}, yielding many interesting phenomena. In particular, Duan, Severini, and Winter~\cite{Duan2015} demonstrated the importance of conclusive exclusion from a list of options by establishing a necessary and sufficient condition for the positive activated feedback-assisted zero-error capacity of a quantum channel.
\update{The work of Duan, Severini, and Winter~\cite{Duan2015} involves a far-reaching extension of the conclusive exclusion of quantum states first explored in the PBR theorem and provides information-theoretic motivations to further study this topic.}

In this paper, we investigate the fundamental limits of conclusive exclusion of quantum states with group actions and explore their applications in quantum foundations and quantum zero-error communication theory. By exploiting the symmetries present in the system and utilizing group representation theory, we derive a simple sufficient condition for conclusive state exclusion under general finite group and compact Lie group actions in terms of the amplitudes of the seed state and the group structure parameters. As a corollary, we further establish the necessary and sufficient condition for conclusive state exclusion under Abelian group actions and generalize the PBR result~\cite{pusey2012reality} to a wider range of settings. Finally, we demonstrate a lower bound on the feedback-assisted and no-signalling-assisted zero-error capacity of classical-quantum channels generated by group actions. These results provide insights into the fundamental limits of information extraction in quantum systems, contribute to the foundations of quantum mechanics, and have potential applications in quantum information processing tasks such as zero-error communication.

\section{Problem Formulation}\label{sec:problem_formulation}
Let $G$ be a compact Lie group or finite group with $|G|$ elements and $\mathbf{R}(G):=\{U_g|g\in G\}$ be the unitary action. \update{For the continuous group, $G$, we simply replace integrals with sums, and $dg$ with $1/|G|$.} The unitary action $\mathbf{R}(G)$ can be seen as a set of black boxes that perform unknown unitary transforms on the system $\cH$.

\update{We focus on the task of conclusive single-state exclusion~\cite{pusey2012reality}, that is, to perfectly exclude any single-state from a given state ensemble $\cE:=\{(p_j,\rho_j)\}_{j=1}^n$, where state $\rho_j$ occurs with probability $p_j$. Assuming equal prior probabilities for each state, we omit probabilities and focus directly on the set of states. The strategies of conclusive exclusion of the state set can be described by the Positive Operator Valued Measures (POVM) $\{M_j\}$. In specific,  call the prepared state $\sigma$, our goal is to measure on $\sigma$ such that, based on the outcome, one can state that $\sigma\neq\rho_j$ for some $j\in\{1,\cdots,n\}$, which is a more general and less ambitious task compared to state discrimination. Then, the objective of obtaining the optimal strategy for the single-state exclusion can be formulated as a semidefinite program (SDP)~\cite{bandyopadhyay2014conclusive}:
\begin{equation}\label{sdp:single_exclusion}
    \begin{aligned}
        \min&\quad\alpha:=\sum_{j=1}^n\tr[\rho_jM_j]\\
        \text{s.t.}&\quad\sum_{j=1}^nM_j=I,\\
        &\quad M_j\geq0,\,\forall j.
    \end{aligned}
\end{equation}}

\update{In fact, for the quantum state $\sigma$ to be excluded, $\tr[\sigma M_j]$ represents the probability of $\sigma\neq \rho_j$. Clearly, to avoid erroneous judgments (i.e., conclusive exclusion), we need to find an exclusion strategy ${M_j}$ such that $\tr[\rho_jM_j]=0$. Conceptually, if $\alpha=0$, we say that conclusive single-state exclusion of the set of states $\cE$ can be done; otherwise, conclusive exclusion is deemed impossible.}

\update{In this work, we consider the structured ensemble, which is generated by the equiprobable operation of a group action $\textbf{R}(G)$ on the initial quantum state $\ket{\psi}$. Formally, we disregard the prior probabilities and denote the set of states to be excluded as set $\cE_\psi:=\{U_g\ket{\psi}|U_g\in\textbf{R}(G)\}$, where $\ket{\psi}$ is also referred to as a seed pure state and the state set $\cE_\psi$ is termed the group orbit within the context of group theory. We aim to find the optimal measurement $\{M_g\}$ such that $\tr[M_g\ketbra{u_g}{u_g}]=0$ for all $g\in G$, where $\ket{u_g}:=U_g\ket{\psi}$. There is a related fact as follows:}

\begin{lemma}[~\cite{bandyopadhyay2014conclusive}]\label{fact:check_condition}
Let $\{M_g|g\in G\}$ be the POVM. Conclusive single-state exclusion of the group-orbit $\cE_\psi$ can be done if and only if the operator $N:=\sum_{g\in G}\ketbra{u_g}{u_g}M_g$ is Hermitian and satisfies $N\leq\ketbra{u_g}{u_g}$ for all $g\in G$.
\end{lemma}

More generally, we introduce the reference system $\cR$ and suppose the input state $\ket{\psi}_{\cH\cR}$ belongs to the tensor product system $\cH\otimes \cR$. Then, our focus is the following orbit:\update{
\begin{equation}
    \begin{aligned}
        \cE_\psi^{\cH\cR}:=\{(U_g\otimes I_\cR)\ket{\psi}_{\cH\cR}|g\in G\},
    \end{aligned}
\end{equation}
where $I_\cR$ is the trivial action on the reference system $\cR$.}

Considering that quantum states are generated from group orbits and group symmetry naturally emerges, we introduce group representation theory as the main analytical tool. Specifically, we denote the set of irreducible representations by $\hat{G}$. Then, due to the isotypical decomposition~\cite{goodman2009symmetry}, the system $\cH$ can be written as:
\begin{equation}
    \begin{aligned}
        \cH\cong \bigoplus_{\mu\in \hat{G}} \cH_\mu\otimes \mathbb{C}^{m_\mu},
    \end{aligned}
\end{equation}
where $m_\mu$ means the multiplicity of the representation space $\cH_\mu$. We further denote the dimension of the space $\cH_\mu$ by $d_\mu$. \update{Notice that one can choose the dimension of the reference system satisfying $d_\cR\geq\max_{\mu\in\hat{G}}d_\mu/m_\mu$. Without loss of generality, we asssume $m_\mu=d_\mu$ for all $\mu\in\hat{G}$. As a result, the whole system has the following decomposition~\cite{hayashi2025indefinite,bisio2010optimal}:}
\begin{equation}
    \begin{aligned}
        \cH\otimes\cR\cong\bigoplus_{\mu\in \hat{G}} \cH_\mu\otimes \mathbb{C}^{d_\mu}.
    \end{aligned}
\end{equation}
We denote $\{\ket{\psi_j^\mu}|j=1,\cdots,d_\mu\}$ and $\{\ket{\phi_j^\mu}|j=1,\cdots,d_\mu\}$ are bases for $\cH_\mu$ and $\mathbb{C}^{d_\mu}$, respectively. The maximally entangled vector between $\cH_\mu$ and $\mathbb{C}^{d_\mu}$ is defined as
\begin{equation}
    \begin{aligned}
        \dket{V_\mu}_{\cH_\mu\mathbb{C}^{d_\mu}}:=\sum_{j=1}^{d_\mu}\ket{\psi_j^\mu}\ket{\phi_j^\mu},
    \end{aligned}
\end{equation}
\update{where the notation $\dket{\cdot}$ denotes the row vectorization of unitary $V_\mu:=\sum_{j=1}^{d_\mu}\ketbra{\psi_j^\mu}{\phi_j^\mu}$ on the Hilbert space $\cH_\mu\otimes\mathbb{C}^{d_\mu}$, i.e., $\dket{V_\mu}=\sum_{j,k=1}^{d_\mu}\bra{\psi_j^\mu}V_\mu\ket{\phi_k^\mu}\ket{\psi_j^\mu}\ket{\phi_k^\mu}$.} More elements can be found in the appendix.~\ref{sec:pre-rep}.

In the following, we are going to investigate the necessary and sufficient conditions to perfectly exclude a single state from the state set under different group actions, especially the specific form of the seed pure state and POVM, and explore its applications in quantum communication.

\section{Main results}\label{sec:main_results}
\subsection{Finite group and compact Lie group}\label{sec:general_group}
In this section, we focus on the group action $\mathbf{R}(G)$ of a finite group $G$. It allows us to analyze the condition of exclusive exclusion in broader contexts. 
\update{Notice that our result also holds for any compact Lie group because our derivation relies on the isotypical decomposition of the Hilbert space, a mathematical structure that is equally applicable to the representations of compact Lie groups.}

\update{For the case of excluding the state in a group-orbit $\cE_\psi^{\cH\cR}$,} we simplify the original SDP in Eq.~\eqref{sdp:single_exclusion} to the following form:
\begin{equation}\label{eq:sdp_primal}
    \begin{aligned}
        \min&\quad\tr[\ketbra{\psi}{\psi}M]\\
        \text{s.t.}&\quad M\geq 0\\
        &\quad\tr_{\cH_\mu}[P_\mu M P_\mu]=d_\mu I_{d_\mu},\quad \forall \mu\in \hat{G}
    \end{aligned}
\end{equation}
where $P_\mu$ denotes the projector on space $\cH_\mu\otimes \mathbb{C}^{d_\mu}$. \update{This implies that whether the conclusive single-state exclusion can be done depends on the choice of the initial pure state $\ket{\psi}$ and the positive operator $M$. Therefore, in order to perfectly exclude the state set $\cE_\psi^{\cH\cR}$, we need to construct a semidefinite operator that satisfies the SDP constraint in Eq.~\eqref{eq:sdp_primal} and possesses a zero eigenvalue. Specifically, we provide it as follows:
\begin{equation}
    \begin{aligned}
        M:=\ketbra{\phi}{\phi},\quad\ket{\phi}:=\frac{1}{\sqrt{|G|}}\bigoplus_{\mu}\sqrt{d_\mu}\dket{V_\mu}.
    \end{aligned}
\end{equation}
Since the construction of pure states forming operator $M$ is not unique, this precludes us from identifying the necessary and sufficient with respect to the seed pure states and measurement operators conditions for the conclusive single-state exclusion task.} Therefore, sufficient conditions from the perspective of operational significance are of great significance.

\begin{theorem}\label{thm:general_case}
    For a given finite group $G$ and its unitary action $\mathbf{R}(G):=\{U_g|g\in G\}$, conclusive single-state exclusion of the group-orbit \update{$\cE_\psi^{\cH\cR}:=\{(U_g\otimes I_\cR)\ket{\psi}_{\cH\cR}|g\in G\}$} can be done if the following inequality holds,
    \begin{equation}\label{eq:if_condition}
        \begin{aligned}
        d_{\mu_0}|a_{\mu_0}|\leq \sum_{\mu\in \hat{G},\mu \neq \mu_0}d_\mu |a_\mu|,
        \end{aligned}
    \end{equation}
    where $\ket{\psi}_{\cH\cR}\in\cH\otimes \cR$ is written as a direct sum of the maximally entangled state between irreducible space $\cH_\mu$ and multiplicity space $\mathbb{C}^{d_\mu}$, i.e., $\ket{\psi}_{\cH\cR}:=\bigoplus_{\mu\in\hat{G}}\frac{a_{\mu}}{\sqrt{d_\mu}}\dket{V_\mu}_{\cH_\mu \mathbb{C}^{d_\mu}}$. The amplitudes satisfy that $d_{\mu_0}|a_{\mu_0}|\geq d_{\mu_1}|a_{\mu_1}| \geq \cdots \geq d_{\mu_j}|a_{\mu_j}|\geq \cdots \geq 0$.
\end{theorem}

\begin{proof}
We are going to demonstrate that there exists a POVM $\{M_g|g\in G\}$ such that $\tr[M_g\ketbra{u_g}{u_g}]=0$, $\ket{u_g}:=(U_g\otimes I_\cR)\ket{\psi}$ (subscripts $\cH\cR$ are omitted for short) such that the sufficient condition in~\eqref{eq:if_condition} holds.

Denote $t:=d_{\mu_0}|a_{\mu_0}|-\sum_{\mu \neq \mu_0} d_{\mu} |a_{\mu}|$. Since $t\leq 0$, there exists $\{\theta_\mu\}$ such that $\sum_\mu d_{\mu}a_{\mu}=0$, where $\theta_\mu$ satisfies $a_\mu=|a_\mu|e^{i\theta_{\mu}}$. We choose the POVM $\{M_g|g\in G\}$, where $M_g=(U_g\otimes I)M_e(U_g^\dagger\otimes I)$ and $M_e:=\frac{1}{|G|}\ketbra{\phi}{\phi}$ for $\ket{\phi}:=\bigoplus_{\mu\in \hat{G}}\sqrt{d_\mu}\dket{V_\mu}_{\cH_\mu \mathbb{C}^{d_\mu}}$. One can verify that $M_e\geq 0$ and the completeness constraint, i.e.,
\begin{equation}\label{eq:completeness constraint}
    \begin{aligned}
        M_G:=\sum_{g\in G}(U_g\otimes I_\cR)M_e(U_g^\dagger\otimes I_\cR)=I_{\cH\cR},
    \end{aligned}
\end{equation}
where Eq.~\eqref{eq:completeness constraint} follows from the fact that the average action on the operator $M_e$ commutes with each action $U_g\otimes I_\cR$, i.e., $[M_G,U_g\otimes I_\cR]=0$. Specifically, it can be written in the following form in terms of Schur's lemma~\cite{bisio2016quantum}:
\begin{equation}\label{eq:schur lemma}
    \begin{aligned}
        M_G=\bigoplus_{\mu\in\hat{G}}I_{\cH_\mu}\otimes \frac{\tr_{\cH_\mu}[P_\mu M_e P_\mu]}{d_\mu},
    \end{aligned}
\end{equation}
where $P_\mu$ denotes the projector on $\cH_\mu\otimes \mathbb{C}^{d_\mu}$ and $\tr_{\cH_\mu}$ is the partial trace on system $\cH_\mu$. It is straightforward to check that $\tr[M_g\ketbra{u_g}{u_g}]=0$, for all $g\in G$, which means that one can use the POVM $\{M_g|g\in G\}$ to complete exclusion perfectly.
\end{proof}

\update{By using an example, we demonstrate that our sufficient condition Eq.~\eqref{eq:if_condition} is not always necessary for non-Abelian groups. We use the group $\textbf{SU}(3)$.}

\begin{example}\label{ex:u3}
    Let $G=\textbf{SU}(3)$ be a unitary group and $\cH=(\mathbb{C}^3)^{\otimes 3}$ be a tensor space of qutrit systems. Suppose the input state $\ket{\psi}_{\cH\cR}$ has the following decomposition:
    \begin{equation}\label{eq:input_state_u3}
        \begin{aligned}
            \frac{a_{\yng(3)}}{\sqrt{d_{\yng(3)}}}\dket{V_{\yng(3)}}\oplus \frac{a_{\yng(2,1)}}{\sqrt{d_{\yng(2,1)}}}\dket{V_{\yng(2,1)}}\oplus \frac{a_{\yng(1,1,1)}}{\sqrt{d_{\yng(1,1,1)}}}\dket{V_{\yng(1,1,1)}}.
        \end{aligned}
    \end{equation}
     \update{First, we choose a set of state amplitudes that satisfies the condition in Eq~\eqref{eq:if_condition}. As depicted in Fig.~\ref{fig:special_case_u3} (left), this allows the corresponding vectors $\{a_\mu d_\mu\}$ to sum to zero with appropriate phases, which in turn allows for the construction of an orthogonal measurement state $\ket{\phi}$ that makes conclusive exclusion possible.
     Specifically, we set up the modulus of amplitudes as $|a_{\yng(3)}|=\frac{4}{\sqrt{1641}}$, $|a_{\yng(2,1)}|=\frac{5}{\sqrt{1641}}$ and $|a_{\yng(1,1,1)}|=\frac{40}{\sqrt{1641}}$, then, the sufficient condition in Eq.~\eqref{eq:if_condition} is satisfied, and one can construct the POVM as $M_e:=\frac{1}{|G|}\ketbra{\phi}{\phi}$,
    \begin{equation}
        \begin{aligned}
            \ket{\phi}:=\sqrt{d_{\yng(3)}}\dket{V_{\yng(3)}}\oplus \sqrt{d_{\yng(2,1)}}\dket{V_{\yng(2,1)}}\oplus \sqrt{d_{\yng(1,1,1)}}\dket{V_{\yng(1,1,1)}}.
        \end{aligned}
    \end{equation}}

    \update{Second, we choose different amplitudes that violate Eq.~\eqref{eq:if_condition}, as shown in Fig.~\ref{fig:special_case_u3} (right). Even so, conclusive exclusion remains possible. In specific, we set up the amplitudes as $a_{\yng(3)}=\frac{1}{\sqrt{2}}$, $|a_{\yng(2,1)}|=\frac{1}{\sqrt{2}}$ and $|a_{\yng(1,1,1)}|=0$. One can still achieve conclusive exclusion of the state in $\cE^{\cH\cR}_\psi$ by employing the following POVM ($M_e^\prime:=\frac{1}{|G|}\ketbra{\phi^\prime}{\phi^\prime}$):}
    \begin{equation}
        \begin{aligned}
            \ket{\phi^\prime}:=\sqrt{d_{\yng(3)}}\dket{V_{\yng(3)}^\prime}\oplus \sqrt{d_{\yng(2,1)}}\dket{V_{\yng(2,1)}^\prime}\oplus \sqrt{d_{\yng(1,1,1)}}\dket{V_{\yng(1,1,1)}^\prime}.
        \end{aligned}
    \end{equation}
    Taking the Young Diagram $\yng(3)$ for instance, there are $d_{\yng(3)}^2$ maximally entangled vectors between spaces$\cH_{\yng(3)}\otimes\mathbb{C}^{d_{\yng(3)}}$. $\dket{V_{\yng(3)}^\prime}$ is one of those, such that $\dbraket{V_{\yng(3)}^\prime}{V_{\yng(3)}}=0$, which is determined by a Heisenberg-Weyl operator~\cite{khatri2020principles} $W_{z,x}$ for $z\neq0,x\neq0$ and $x,z\in\{0,\cdots,d_{\yng(3)}-1\}$.
\end{example}

\begin{figure}[H]
    \centering
    \includegraphics[width=0.9\linewidth]{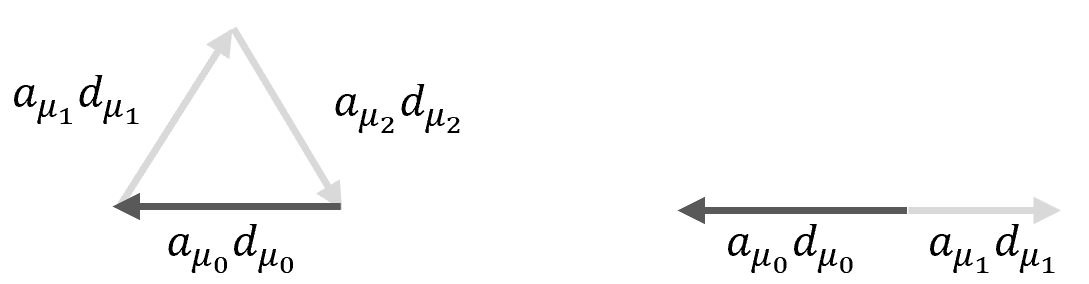}
    \caption{The seed pure states with the form like Eq.~\eqref{eq:input_state_u3} that can be used to perfectly exclude states in $\cE^{\cH\cR}_\psi$. Each arrow in this two diagrams represents a vector $a_\mu d_\mu=|a_\mu|d_\mu e^{i\theta_\mu}.$}
    \label{fig:special_case_u3}
\end{figure}
\subsection{Generalized PBR Game}\label{sec:pbr_game}
Suppose $G$ is an Abelian group. It has one-dimensional irreducible representations, i.e., $d_\mu=1$ and the reference system $\cR$ is not required. Based on the characteristic of the group structure, we find that the necessary condition for general groups can be strengthened to a necessary and sufficient condition of exclusive single-state exclusion, which is only determined by the seed pure state.

\begin{proposition}\label{thm:abelian}
    Given an Abelian group $G$ and its unitary action $\mathbf{R}(G):=\{U_g|g\in G\}$, the conclusive single-state exclusion of the group-orbit \update{$\cE^{\cH}_\psi:=\{U_g\ket{\psi}|g\in G\}$} can be done if and only if 
    \begin{equation}\label{eq:iff_condition_abelian}
        \begin{aligned}
            |a_{\mu_0}|\leq \sum_{\mu\in \hat{G},\mu \neq \mu_0}|a_\mu|,
        \end{aligned}
    \end{equation}
    where $\ket{\psi}:=\sum_{\mu}a_\mu\ket{v_\mu}$ and $\{\ket{v_\mu}\}$ denotes the \update{basis} of $\cH$. The amplitudes satisfy that $|a_{\mu_0}|\geq |a_{\mu_1}| \geq \cdots \geq |a_{\mu_j}|\geq \cdots \geq 0$.
\end{proposition}
\begin{proof}
We are going to demonstrate that there exists a POVM $\{M_g|g\in G\}$ such that $\tr[M_g\ketbra{u_g}{u_g}]=0$, $\ket{u_g}:=U_g\ket{\psi}$ such that the condition in~\eqref{eq:iff_condition_abelian} holds. We denote $t:=|a_{\mu_0}|-\sum_{\mu\neq \mu_0}|a_\mu|$.

On the one hand, since the dimension of each irreducible representation of an Abelian group equals to one, i.e., $d_\mu=1$, for all $\mu\in\hat{G}$, the sufficient condition in Eq.~\eqref{eq:if_condition} can be reduced to the case of Eq.~\eqref{eq:iff_condition_abelian}. According to Theorem~\ref{thm:general_case}, if $t\leq 0$, the conclusive single-state exclusion of the states in orbit $\cE^{\cH}_\psi$ can be done.

On the other hand, we consider the case of $t>0$. Different from Example~\ref{ex:u3}, the space corresponding to the maximum amplitude is non-degenerate under the action of the Abelian group. This makes it impossible for conclusive exclusion if the conditions are violated $(t>0)$. The main idea of the following proof is to use the dual SDP of conclusive state exclusion~\cite{bandyopadhyay2014conclusive}. Specifically, we denote a state as follows:
\begin{equation}
    \begin{aligned}
        \ket{\phi^\prime}:=\frac{1}{\sqrt{|G|}}(e^{i\theta_{\mu_0}}\ket{v_{\mu_0}}-\sum_{\mu\neq{\mu_0}}e^{i\theta_{\mu}}\ket{v_{\mu}}).
    \end{aligned}
\end{equation}
Denote $M_e:=\ketbra{\phi^\prime}{\phi^\prime}$, then the POVM is given by $\{M_g^\prime:=U_gM_eU_g^\dagger|g\in G\}$. According to lemma~\ref{fact:check_condition}, to make sure that it is an optimal measurement, $N$ should be Hermitian and satisfy that $N\leq \ketbra{u_g}{u_g}$ for all $g\in G$, where $N:=\sum_{g\in G}\ketbra{u_g}{u_g}M_g^\prime$. Under the basis $\{v_\mu\}$, we further rewrite it as follows:
\begin{equation}
    \begin{aligned}
        N=t(|a_{\mu_0}|\cdot\ketbra{v_{\mu_0}}{v_{\mu_0}}-\sum_{\mu\neq\mu_0}|a_\mu|\cdot\ketbra{v_\mu}{v_\mu}).
    \end{aligned}
\end{equation}
Let $A_g:=-N+\ketbra{u_g}{u_g}$, note that $A_g=U_gA_eU_g^\dagger$, we only need to consider the eigenvalues of $A_e$. Assume that the eigenvalues of $A_e$ are $\lambda_0({A_e})\leq\lambda_1({A_e})\leq\cdots\leq\lambda_{|G|-1}(A_e)$. Then, by Weyl’s inequality,
\begin{equation}
    \begin{aligned}
        \lambda_k(A_e)\geq\lambda_k(-N)+\lambda_0(\ketbra{\psi}{\psi}),
    \end{aligned}
\end{equation}
we have $\lambda_k(A_e)>0$, for $k=1,\cdots,|G|-1$, which means that $A_e$ has at most one non-positive eigenvalue $\lambda_0(A_e)$. Notice that
\begin{equation}
    \begin{aligned}
        A_e\ket{\phi^\prime}&=t(|a_{\mu_0}|\ketbra{v_{\mu_0}}{v_{\mu_0}}-\sum_{\mu\neq\mu_0}|a_\mu|\ketbra{v_\mu}{v_\mu})\ket{\phi^\prime}+\ket{\psi}\langle\psi|\phi^\prime\rangle\\
        &=\frac{1}{\sqrt{|G|}}(-t\sum_{\mu}a_\mu\ket{v_\mu}+t\ket{\psi})=0,
    \end{aligned}
\end{equation}
then $\lambda_0(A_e)=0$, which means that $N\leq \ketbra{v_g}{v_g}$. Therefore, based on lemma~\ref{fact:check_condition}, $\{M_g^\prime|g\in G\}$ is an optimal measurement. Also, $\tr[N]=t^2$ is strictly positive in this case, and this is the optimal error probability we can achieve when we cannot do the exclusion conclusively.
\end{proof}

Proposition~\ref{thm:abelian} allows us to investigate a kind of specific Abelian group, which plays a main role in the PBR theorem~\cite{pusey2012reality}. Let us first recall the PBR theorem. It is a foundational result in quantum mechanics that addresses the ontological status of the quantum state. Specifically, it aims to determine whether the quantum state is ontic (a physical property of a system) or epistemic (merely information about the system's state). In 2012, Pusey, Barrett, and Rudolph (abbreviated as PBR) proved that quantum states must describe reality directly. This result challenges epistemic interpretations of quantum mechanics, which view the quantum state as representing an observer's knowledge or information about a system. Formally, it is related to the quantum state exclusion task.


By analyzing the natural action $\mathbf{R}(G):=\{g|g\in G\}$ of the Abelian group $G:=\{Z,I\}^{\otimes n}$, we can recover the necessary and sufficient conditions proposed in the PBR theorem. It simplifies the proof of the PBR theorem from the perspective of group action.
\begin{corollary}[Recover the condition 
 of PBR game~\cite{pusey2012reality}]
    For the PBR game, the conclusive single-state exclusion of the state set $\cE_\psi^{\cH}:=\{U_{\vec{x}}\ket{\psi_{\vec{x}}}|\vec{x}\in\{0,1\}^n\}$ can be done if and only if
    \begin{equation}
        \begin{aligned}
            \left(1+\tan\frac{\theta}{2}\right)^{n}\geq 2,
        \end{aligned}
    \end{equation}
    where $\ket{\psi_0}=\cos\frac{\theta}{2}\ket{0}+\sin\frac{\theta}{2}\ket{1}$, $\ket{\psi_1}=\cos\frac{\theta}{2}\ket{0}-\sin\frac{\theta}{2}\ket{1}$, $\theta\in[0,\frac{\pi}{2}]$, and $U_{\vec{x}}:=\otimes_{j=1}^n Z^{x_j}$.
\end{corollary}
\begin{proof}
    Denote $\ket{\psi}:=\ket{\psi_0}^{\otimes n}$. By definition, we have $\cos\frac{\theta}{2}\geq\sin\frac{\theta}{2}\geq 0$, which satisfies the conditions of amplitudes in the computational basis. Then, by Proposition~\ref{thm:abelian}, conclusive single-state exclusion can be done if and only if
    \begin{equation}
        \begin{aligned}
            (\cos\frac{\theta}{2})^n\leq(\cos\frac{\theta}{2}+\sin\frac{\theta}{2})^n-(\cos\frac{\theta}{2})^n.
        \end{aligned}
    \end{equation}
    Therefore, it is equivalent to $(1+\tan\frac{\theta
    }{2})^n\geq 2$, which completes this proof.
\end{proof}

As shown above, we have a symmetry perspective to understand the PBR theorem. Notably, the results related to Abelian groups or general groups could enable a natural extension of the PBR game to higher-dimensional quantum systems or more general scenarios. This is helpful for the exploration of fundamental properties inherent in quantum mechanics.

\subsection{Zero-error Communication}\label{sec:zero_error_communication}
For a given finite group $G$, we have the state set, i.e., the group-orbit $\cE_\psi^{\cH}$ with respect to certain seed pure state $\ket{\psi}$. Let us consider the communication model, $ \cN: $g$ \to \ket{u_g}$, in which the sender, Alice, transmits a quantum state $\ket{u_g}\in \cE_\psi^{\cH}$ encoded with classical information $g\in G$ to Bob, who retrieves the information through quantum measurement $\{M_g\}$ at the receiving end as shown in Fig.~\ref{fig:zero_error_model}. \update{Then the key question is to understand how much information this channel can convey perfectly on average in the asymptotic regime.}

\begin{figure}[htbp]
    \centering
    \includegraphics[width=0.75\linewidth]{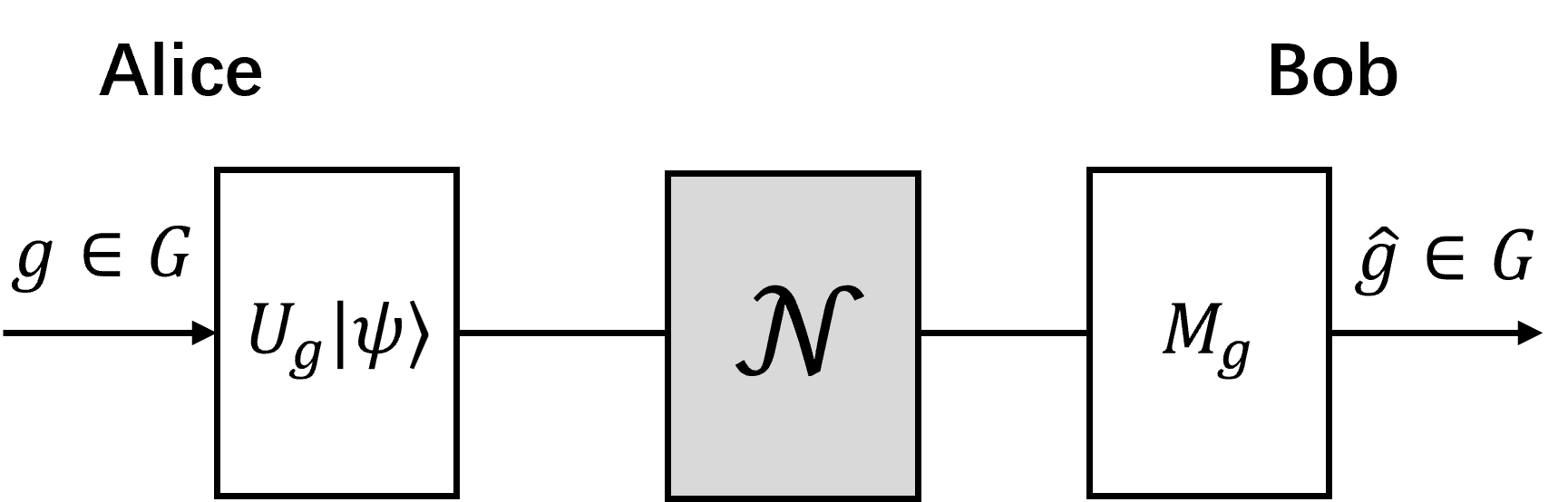}
    \caption{Classical communication about finite group information.}
    \label{fig:zero_error_model}
\end{figure}

Notice that the classical-quantum channel $\cN$ in this communication model can be characterized by a bipartite graph $\Gamma$, and the classical capacity can be obtained by the fractional packing number of the bipartite graph~\cite{shannon1956zero,cubitt2011superactivation}. Details are provided in the appendix~\ref{apx:zero_error_communication}. Now, we are ready to build the connection between zero-error communication and quantum state exclusion. 

Specifically, we provide a perspective from the conclusive single-state exclusion to measure the capacity of transmitting information as shown in the following result.

\begin{proposition}\label{thm:zero error capacity}
For a classical-quantum channel 
$$\cN: g \to \ket{u_g},$$ if Eq.~\eqref{eq:if_condition} is satisfied, then the feedback-assisted zero-error capacity as well as the classical non-signaling-assisted zero-error capacity of $\cN$ possess the following lower bound:
\begin{align}
    C_{0,F}(\cN) = C_{0,NS}(\cN) \ge \log\frac{|G|}{|G|-1}.
\end{align}
\end{proposition}

\begin{proof}
We can construct a classical channel with a bipartite graph $\Gamma(h|g)=\tr(M_h\ketbra{v_g}{v_g})$, where $\Gamma(g|g)=0$, and the no-signalling assisted zero-error capacity of this channel is its fractional packing number, which is at least $\frac{|G|}{|G|-1}$. Then, $C_{0,F}(\cN) = C_{0,NS}(\cN) \ge \log\frac{|G|}{|G|-1}>1$. 
\end{proof}

\update{This result connects the physical possibility of state exclusion to the information-theoretic concept of zero-error capacity. It's important to note that conclusive state exclusion is a fundamentally different task from state discrimination; it allows a receiver to determine with certainty which state was not sent, rather than which state was sent. The success of conclusive exclusion guarantees that some information can be extracted with perfect certainty, which is the necessary condition for a channel to have a positive zero-error capacity when assisted by resources like feedback or non-signalling correlations.}
Moreover, a larger group size decreases the capacity lower bound, since conclusive exclusion for a larger set reveals less information about the system. We also note that conclusive state exclusion of states generated by multiple channel uses was studied in \cite{Duan2015} to establish necessary and sufficient conditions for positive feedback-assisted capacity when allowing for a constant amount of activating noiseless communication.
\section{Conclusions and discussions}\label{sec:conclusion}

This work establishes explicit criteria for conclusive single-state exclusion under group symmetry constraints, with implications for generalizing the PBR result in quantum foundations and estimating the zero-error capacity of classical-quantum channels. For general symmetries governed by finite groups (or compact Lie groups), conclusive exclusion becomes achievable when the seed pure state is constructed through the superposition of maximally entangled states between the irreducible representation spaces and the corresponding multiplicity spaces, governed by the group structure. For Abelian groups, building on the general sufficient condition, we further demonstrate a necessary and sufficient condition related to the seed pure state, showing that exclusion feasibility reduces to balancing the maximum amplitude against the cumulative summation of remaining amplitudes. This leads to a simple proof of the PBR result~\cite{pusey2012reality} and generalizes it to higher-dimensional quantum systems through amplitude threshold conditions. We have also demonstrated that these results bridge abstract symmetry principles to operational communication models. We have established a lower bound on the feedback-assisted or non-signalling-assisted zero-error capacity of classical-quantum channels characterized by group actions, revealing the fundamental achievability of transmitting information via conclusive state exclusion. Our results may also shed light for the quantum hypothesis exclusion~\cite{Ji2025}.

Future work should extend these criteria to mixed-state ensembles to address decoherence effects for practical quantum information processing. It will also be interesting to explore the conclusive exclusion of quantum operations. Moreover, it remains an open challenge to develop more sophisticated and detailed characterizations of information or communication capacity based on multi-state conclusive exclusion.

\vspace{2mm}
\textit{Note added.} 
\update{While finishing this manuscript, we became aware of a closely related work~\cite{Arnau2025qse} that independently investigated the quantum state exclusion under group action settings. They gave a necessary and sufficient condition in terms of the Gram matrix $G$, in which one needs to calculate the eigenvalues of $G$ to determine whether the state set can be perfectly excluded. Our work, in contrast, focuses on providing explicit and operational conditions derived from group representation theory, such as those concerning the specific form of the seed pure state and POVMs. This approach is intended to help better understand the fundamental limits of information extraction by revealing the required physical intuition, for instance, the constraints on the seed state's amplitudes. It also allows us to further explore applications in quantum foundations, such as generalizing the PBR result, and in quantum communication, by establishing a lower bound for the zero-error capacity of group-generated channels, which are distinct contributions of our manuscript.}
\section*{Acknowledgement}
\update{We sincerely thank the anonymous referee of Physical Review A for very helpful comments that helped us improve the manuscript.}
We would like to thank Yin Mo and Chengkai Zhu for their helpful comments. XW would like to thank Runyao Duan for helpful discussions on conclusive state exclusion of the Abelian group. This work was partially supported by the National Key R\&D Program of China (Grant No.~2024YFB4504004), the Guangdong Provincial Quantum Science Strategic Initiative (Grant No.~GDZX2403008, GDZX2403001), the Guangdong Provincial Key Lab of Integrated Communication, Sensing and Computation for Ubiquitous Internet of Things (Grant No.~2023B1212010007), the Quantum Science Center of Guangdong-Hong Kong-Macao Greater Bay Area, and the Education Bureau of Guangzhou Municipality.

\bibliography{ref}


\onecolumngrid
\appendix
\setcounter{subsection}{0}
\setcounter{table}{0}
\setcounter{figure}{0}

\vspace{1cm}

\renewcommand{\theequation}{S\arabic{equation}}
\renewcommand{\thesubsection}{\normalsize{\arabic{subsection}}}
\renewcommand{\theproposition}{S\arabic{proposition}}
\renewcommand{\thedefinition}{S\arabic{definition}}
\renewcommand{\thefigure}{S\arabic{figure}}
\setcounter{equation}{0}
\setcounter{table}{0}
\setcounter{section}{0}
\setcounter{proposition}{0}
\setcounter{definition}{0}
\setcounter{figure}{0}

\section{Elements of Representation Theory}\label{sec:pre-rep}
\begin{definition}[Unitary Representation]
    Let $G$ be a group and $\cH$ a Hilbert space. The unitary representation of $G$ on $\cH$ is a group homomorphism $\mathbf{R}:g\to U_g$, where $U_g$ is unitary for all $g\in G$.
\end{definition}
Notice that a group action describes how a group $G$ permutes elements of a set $X$, while a group representation is a specific type of action in which $G$ acts linearly on a Hilbert space $\cH$. The group action mentioned in this work refers to the unitary representation of groups.

\begin{lemma}[Schur lemma]
    Let $U_g$ and $V_g$ be irreducible representations of the group $G$ on Hilbert spaces $\cH$ and $\cK$, respectively. Let $O:\cH\to\cK$ be an operator such that $OU_g=V_gO$ for all $g\in G$. If $U_g$ and $V_g$ are equivalent representations, then $O=\lambda I$, where $I$ is an identity operator; otherwise, $O=0$.
\end{lemma}

The Schur lemma is a building block for investigating the general form of an operator commuting with a group representation.

\begin{theorem}[Commutant]
    Let $U_g$ be a unitary representation of a group $G$ and $O$ be an operator satisfying that $[O,U_g]=0$ for all $g\in G$. Then, we have
    \begin{equation}\label{eq:iso_decompo}
        \begin{aligned}
            O=\bigoplus_{\mu \in \text{Irr}(G)} I_{\mu}\otimes O_\mu,
        \end{aligned}
    \end{equation}
    where $O_\mu$ denotes an operator on the multiplicity space of the irreducible representation $U_g^{\mu}$.
\end{theorem}
Therefore, it is straightforward to see that the group average operator $\overline{O}:=\frac{1}{|G|}\sum_{g\in G}U_g O U_g^\dagger$ can be decomposed into the form like Eq.~\eqref{eq:iso_decompo}.

The direct sum of irreducible representations leads to the decomposition of the Hilbert space $\cH$, i.e., isotypical decomposition. Specifically, suppose $G=\textbf{SU}(d)$ is a unitary group, and the carry space is $(\mathbb{C}^d)^{\otimes n}$. Then, we have
\begin{equation}
    \begin{aligned}
        (\mathbb{C}^d)^{\otimes n}=\bigoplus_{\mu\in Y_n^d}\cH_\mu\otimes \mathbb{C}^{m_\mu},
    \end{aligned}
\end{equation}
where $Y_n^d$ denotes the set of Young Diagrams. Therefore, for the case of $n=3$ and $d=3$, we have three irreducible representations, denoted by Young Diagrams, $\yng(3)$, $\yng(2,1)$, and $\yng(1,1,1)$, respectively. The dimensions of irreducible representation spaces and the corresponding multiplicity spaces are $d_{\yng(3)}=10$, $d_{\yng(2,1)}=8$, $d_{\yng(1,1,1)}=1$, $m_{\yng(3)}=1$, $m_{\yng(2,1)}=2$, and $m_{\yng(1,1,1)}=1$. Thus, we have 
\begin{equation}
    \begin{aligned}
        (\mathbb{C}^3)^{\otimes 3}=\cH_{\yng(3)}\oplus\cH_{\yng(2,1)}\oplus \cH_{\yng(2,1)}\oplus \cH_{\yng(1,1,1)}.
    \end{aligned}
\end{equation}

\section{Zero-error communication}\label{apx:zero_error_communication}
For a classical channel $\cN$ that have finite input and output alphabets ($X$ and $Y$), we write $\cN(y|x):=P(\text{output}=y | \text{ input}=x)$ as transition probabilities of this channel. The bipartite graph (transition graph) $\Gamma$ of this classical channel $\cN$ is given by the adjacency matrix on $X \times Y $:
 \begin{eqnarray}
 \Gamma(y|x)=
\begin{cases}
1, &\cN(y|x)>0 \cr 0, &\cN(y|x)=0\end{cases}
\end{eqnarray}

The fractional packing number of a bipartite graph $\Gamma(y|x)$ is
\begin{equation}
{\alpha ^*}(\Gamma) = \max \sum\limits_x {v(x)} \ \text{  s.t. }\sum\limits_x {v(x)\left\lceil {\Gamma(y|x)} \right\rceil }  \le 1,v(x) \ge 0,
\end{equation}
which is the maximum total weight allowed in fractional packing of $\Gamma$.
This number has been proved to be the feedback assisted zero-error capacity as well as the classical non-signalling assisted zero-error capacity of a classical channel $\cN$~\cite{shannon1956zero,cubitt2011superactivation}.

\end{document}